\documentclass[runningheads]{llncs}

\usepackage{upgreek}
\usepackage[T1]{fontenc}
\usepackage{amsmath,amssymb,mathtools}
\usepackage{microtype}
\usepackage{algorithm}
\usepackage{algpseudocode}
\usepackage{enumitem}
\usepackage[hidelinks]{hyperref}
\usepackage{todonotes}
\usepackage{hyperref}
\hypersetup{
  colorlinks=true,
  linkcolor=blue,
  urlcolor=blue,
  citecolor=blue
}

\newenvironment{claimproof}{%
  \par\noindent\textit{Proof of the claim.}\ }{%
  \hfill$\triangleleft$\par\medskip}

\usepackage{xcolor}
\hypersetup{
    colorlinks=true,
    linkcolor=blue,
    urlcolor=blue,
    citecolor=blue
}
\usepackage{thmtools,thm-restate}

\newcommand{\vi}{\operatorname{vi}}
\newcommand{\cc}{\operatorname{cc}}

\newcommand{\Oh}{O^*}
\newcommand{\VI}{\textnormal{\textsc{Vertex Integrity}}}

\title{Solving Vertex Integrity Faster than $2^n$}
\titlerunning{ETH Lower Bound and Exact Exponential Algorithm}

\author{
Sandip Das\inst{1} \and
Sweta Das\inst{1} \and
Sk Samim Islam\inst{2} \and
Ritam Manna Mitra\inst{1} \and
Aashirwad Mohapatra\inst{1} \and
Arkaprava Paul\inst{1}
}

\authorrunning{S. Das et al.}

\institute{
Indian Statistical Institute, Kolkata, India\\
\email{
sandip.das.69@gmail.com,
sweta.twin@gmail.com,
rmmitra98@gmail.com,
aashirwad\_r@isical.ac.in,
arkapravapaul@gmail.com
}
\and
University of Electronic Science and Technology of China,
Chengdu, China\\
\email{samimislam08@gmail.com}
}

\begin{document}
\maketitle

\begin{abstract}
Vertex Integrity is a classical and well-studied graph-optimization
problem.
For an undirected graph $G$ and a vertex set $S\subseteq V(G)$, let
$\cc(G \setminus S)$ denote the set of connected components of the induced subgraph $G[V(G)\setminus S]$. The
\emph{vertex integrity} of $G$ is
\[
\vi(G)=\min_{S\subseteq V(G)}
\left(
|S|+\max_{C\in\cc(G \setminus S)}|V(C)|
\right).
\]

The \textsc{vertex integrity} problem asks whether there is a set $S\subseteq V(G)$ such that $ |S|+\max_{C\in\cc(G \setminus S)}|V(C)|\leq k$.
 Clark, Entringer, and Fellows proved in 1987 that Vertex Integrity is NP-complete, even on planar graphs. Their reduction constructs a graph with quadratically many vertices and therefore does not rule out a subexponential-time algorithm in terms of the number of vertices. Several subsequent works studied the complexity of Vertex Integrity and posed related questions on faster and parameterized algorithms. However, whether Vertex Integrity admits a \(2^{o(n)}\)-time algorithm remained unresolved.

We answer this question negatively. We give a polynomial-time reduction with linear vertex blow-up
from \textnormal{\textsc{Vertex Cover}} on subcubic graphs and prove that,
unless the Exponential Time Hypothesis fails, Vertex Integrity has no
$2^{o(n)}$-time algorithm.

On the algorithmic side, the straightforward exact algorithm enumerates all
vertex sets and runs in $O^*(2^n)$ time. We break this $2^n$ barrier and
give a deterministic $O(1.9602^n)$-time exact algorithm.

\end{abstract}

\section{Introduction}

For a finite undirected graph $G$, its \emph{vertex integrity} is
\begin{equation}\label{eq:intro-vi}
\vi(G)=\min_{S\subseteq V(G)}
\left(|S|+\max_{C\in\cc(G \setminus S)}|V(C)|\right),
\end{equation}
where $\cc(G \setminus S)$ is the set of connected components of $G \setminus S$. The \textsc{Vertex Integrity} problem asks, given
$G$ and an integer $k$, whether $\vi(G)\leq k$.

Vertex integrity was introduced as a measure of network vulnerability by
Barefoot, Entringer, and Swart~\cite{BarefootEntringerSwart1987}; see also
the survey of Bagga et al.~\cite{bagga1992survey}. It satisfies
$\operatorname{td}(G)\leq\vi(G)\leq\operatorname{vc}(G)+1$
~\cite{gima2022exploring}, where $\operatorname{td}(G)$ and
$\operatorname{vc}(G)$ denote the treedepth and vertex-cover number of
$G$, respectively.

Clark, Entringer, and Fellows proved that \textsc{Vertex Integrity} is
NP-complete even on planar graphs~\cite{clark1987computational}. It is also
NP-hard on co-bipartite and chordal graphs, planar bipartite graphs of
maximum degree four, and line graphs, whereas polynomial-time algorithms
are known for trees, cactus graphs, split graphs, and some intersection
graph classes~\cite{drange2016computational,gima2025structural}. The
weighted version, in which vertices have positive weights, has also been
studied extensively.

For target value $p$, Fellows and Stueckle gave an $O(p^{3p}n)$-time
algorithm~\cite{fellows1989immersion}. Drange, Dregi, and van~'t~Hof
improved this to $O(p^{p+1}n)$, also for Weighted Vertex Integrity, and
obtained a kernel with at most $p^3$ vertices
~\cite{drange2016computational}. Casel et al. subsequently improved the
unweighted kernel to $3p^2$ vertices~\cite{casel2026combining}. Whether
the problem admits a $c^pn^{O(1)}$-time algorithm for some constant $c$
remains open.

Gima et al.~\cite{gima2025structural} proved that Unweighted Vertex
Integrity is fixed-parameter tractable parameterized by cluster vertex
deletion number and W[2]-hard parameterized by pathwidth. Hanaka et
al.~\cite{hanaka2024parameterized} proved W[1]-hardness parameterized by
treedepth or feedback edge set number and fixed-parameter tractability
parameterized by max-leaf number. They also gave single-exponential
algorithms for Weighted Vertex Integrity parameterized by vertex cover or
modular-width. Gima et al.~\cite{gima2025structural} further obtained an
$O(\log\mathsf{opt})$-approximation using an approximation for component
order connectivity~\cite{lee2017partitioning}; a constant-factor
approximation remains open.

The exact exponential complexity with respect to $n=|V(G)|$ was also
open. Enumerating every $S\subseteq V(G)$ gives an $O^*(2^n)$-time
algorithm, but, to the best of our knowledge, no
$O^*(\alpha^n)$-time algorithm with $\alpha<2$ was known.
There was also a gap on the lower-bound side. The NP-hardness reduction of
Clark, Entringer, and Fellows constructs, from an $n$-vertex source graph,
a graph with $2n^2$ vertices. Therefore, this reduction does not rule out a
$2^{o(n)}$-time algorithm for Vertex Integrity, where $n$ is the number of
vertices of the input graph. Consequently, NP-hardness was known, but it
was still possible that Vertex Integrity could be solved in
subexponential time. Thus, two basic questions concerning the exact
exponential complexity of Vertex Integrity remained open: whether the
problem admits a $2^{o(n)}$-time algorithm, and whether the straightforward
$O^*(2^n)$ running time can be improved.

Throughout the paper, $O^*(\cdot)$ suppresses factors polynomial in the
input size. We use the standard fact that, for fixed constants $1<a<b$ and
$c\geq 0$, we have $n^c a^n=O(b^n)$. Hence, an $O^*(a^n)$-time algorithm
also runs in $O(b^n)$ time for every fixed constant $b>a$.

\subsection{Our Contribution}
First, we establish an ETH lower bound that rules out subexponential-time algorithms with respect to the number of vertices, unless ETH fails.

\begin{restatable}{theorem}{ethlowerbound}\label{thm:eth}
Unless ETH~\footnote{The Exponential Time Hypothesis (ETH) states that there is no
$2^{o(n)}$-time algorithm for 3-SAT on formulas with $n$ variables.} fails,
\VI{} has no $2^{o(n)}n^{O(1)}$-time algorithm on graphs with $n$ vertices.
\end{restatable}

Theorem~\ref{thm:eth} follows by a reduction from
\textnormal{\textsc{Vertex Cover}} on graphs of maximum degree at most three.
Given an $n$-vertex instance, the reduction constructs an equivalent instance
of \VI{} with exactly $5(n+1)+1$ vertices. Thus, the number of vertices
increases only by a constant factor, which gives the stated ETH lower bound.

Our second result gives an exact exponential-time algorithm and breaks the $2^n$ barrier.

\begin{restatable}{theorem}{exactalgorithm}\label{thm:algorithm}
Vertex integrity of an $n$-vertex graph G can be computed in
\begin{equation}\label{eq}
\Oh\left(\left(\frac{5}{2^{2/5}3^{3/5}}\right)^n\right)
=O(1.9602^n)
\end{equation}
time and $O(1.9602^n)$ space. An optimal integrity separator can be
found within the same bounds.
\end{restatable}

\paragraph{Overview of the algorithm.}
Let $S$ be an integrity separator and $C$ be a largest component of
$G \setminus S$,let $U=V\setminus(S\cup V(C))$. Since
$\vi(G)=n-|U|$, we enumerate $S$ or $U$ whenever one has 
$\leq t$ vertices.If both are larger, we obtain an
anticomplete partition $A,B$ of $V\setminus S$ such that
$|A|,|V\setminus N[A]|\leq t$. We enumerate $A$ and use a subset dynamic
program to choose the best $B\subseteq V\setminus N[A]$. Thus, only subsets
of size at most $t$ are considered,taking $t =\lceil2n/5 \rceil$ yields a deterministic
$O(1.9602^n)$-time algorithm.

\subsection{Organization}

In Section~\ref{sec:prelim}, we introduce the notation and preliminary results used in the paper. In Section~\ref{sec:lower}, we prove the ETH lower bound stated in Theorem~\ref{thm:eth}. In Section~\ref{sec:algorithm}, we present the structural results and the exact exponential-time algorithm proving Theorem~\ref{thm:algorithm}. Finally, in Section~\ref{sec:conclusion}, we conclude the paper and mention some directions for future work.

\section{Preliminaries}\label{sec:prelim}

All graphs considered are finite, undirected, and simple. For a graph $G=(V,E)$ and
$X\subseteq V$, let $G[X]$ denote the subgraph induced by $X$ and  $G-X=G[V\setminus X]$. The set of connected components of $G$ is denoted by $\cc(G)$. Define $ \displaystyle \mu_G(X)=\max_{C\in\cc(G[X])}|V(C)|$,$V(C)$ referring to the set of vertices within component $C$ with $\mu_G(\varnothing)=0$. When the graph is fixed, we simply write $\mu(X)$. The \emph{vertex integrity} of a graph $G$ is defined as $\vi(G)=\min_{S\subseteq V}\bigl(|S|+\mu(V\setminus S)\bigr).$

A subset $S \subseteq V$ is called an \emph{integrity separator} if $\vi(G) = |S|+\mu(V\setminus S)$.

For $X\subseteq V$, 
let $N(X)=\{v\in V\setminus X : v\text{ has a neighbour in }X\}$ and let $N[X]=X\cup N(X)$. Two disjoint vertex sets $A,B\subseteq V(G)$ are \emph{anticomplete} to each other if no edge has one endpoint in $A$ and the other in $B$.

\medskip
\noindent
\fbox{%
  \begin{minipage}{\dimexpr\linewidth-2\fboxsep-2\fboxrule}
  \textsc{Vertex Cover} problem
  \begin{itemize}[leftmargin=1.5em]
    \item \textbf{Input:} An undirected graph $G = (V, E)$ and an integer $k \in \mathbb{N}$.
    \item \textbf{Question:} Does there exist a \emph{vertex cover} $S \subseteq V$ of size at most $k$; that is, a set of at most $k$ vertices such that every edge in $E$ has at least one endpoint in $S$?
  \end{itemize}
  \end{minipage}%
}
\medskip

\section{Lower Bound Under ETH (Proof of Theorem~\ref{thm:eth})}\label{sec:lower}

In this section, we show that \VI{} has no
$2^{o(n)}n^{O(1)}$-time algorithm on graphs with $n$ vertices. 
We use the Exponential Time Hypothesis (ETH), which states that there is no
$2^{o(n)}$-time algorithm for 3-SAT on formulas with $n$ variables
\cite{impagliazzo2001complexity}. The following result follows from a result of Bacs\'o et al.~\cite{bacso2019subexponential} which says that, assuming ETH,
\textnormal{\textsc{Maximum Independent Set}} has no $2^{o(n)}$-time
algorithm on $n$-vertex graphs of maximum degree at most three.

\begin{lemma}\label{lem:source}
Unless ETH fails, \textsc{Vertex Cover} on $n$-vertex graphs of maximum
degree at most three cannot be solved in $2^{o(n)}$ time.
\end{lemma}

\begin{proof}
Bacs\'o et al.~\cite{bacso2019subexponential} proved that, unless ETH
fails, \textsc{Maximum Independent Set} on $n$-vertex graphs of maximum
degree at most three cannot be solved in $2^{o(n)}$ time. A set
$I\subseteq V(G)$ is an independent set if and only if
$V(G)\setminus I$ is a vertex cover. Hence $G$ has an independent set
of size at least $\ell$ if and only if it has a vertex cover of size at
most $n-\ell$. The transformation does not change the graph, and
therefore preserves the number of vertices.
\hfill $\square$
\end{proof}

 We reduce the \textsc{Vertex Cover} problem on graphs with maximum degree three to the \textsc{Vertex Integrity} problem.

Let $G_1$ be an $n$-vertex graph of maximum degree at most three. Since every vertex of $G_1$ has degree at most three, a proper $4$-coloring of $G_1$ can be computed in polynomial time. Let $P_1,P_2,P_3,P_4$ be its color classes; thus, each $P_i$ is an independent set. We construct a graph $G_2$ as follows.

For each $i\in\{1,2,3,4\}$, add $n+1-|P_i|$ new vertices to $P_i$ to obtain a set $X_i$ of order $n+1$. Introduce another set $X_5$ consisting of $n+1$ new vertices. Make $G_2[X_i]$ a clique for every $i\in\{1,2,3,4,5\}$. For every pair of distinct indices $i,j\in\{1,2,3,4\}$ and every $u\in P_i$ and $v\in P_j$, add the edge $uv$ to $G_2$ if and only if $uv\in E(G_1)$. No other edges are added between distinct sets among $X_1,\ldots,X_5$.

Finally, introduce a vertex $z$ and make it adjacent to every vertex of $\bigcup_{i=1}^{5}X_i$. The resulting graph $G_2$ is connected and has
\begin{equation*}
|V(G_2)|=5(n+1)+1.
\end{equation*}
Every vertex of $X_i\setminus P_i$, for $i\in\{1,2,3,4\}$, has exactly one neighbour outside $X_i$, namely $z$. Similarly, every vertex of $X_5$ has exactly one neighbour outside $X_5$, namely $z$.
\begin{figure}[h]
    \centering
    \includegraphics[width=\textwidth]{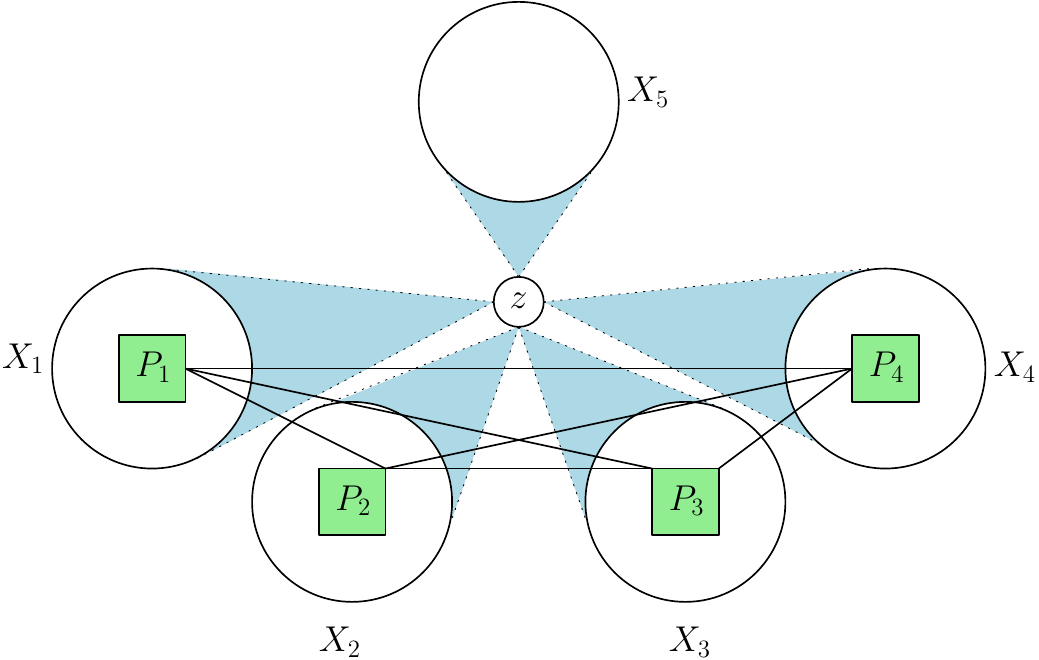}
    \caption{An illustration of the construction of the graph $G_2$ used in the proof of Lemma~\ref{lem:equivalence}. The set $V(G_1) = \bigcup_{i \in [4]} P_i$. For all $i \in \{1,2,3,4\}$ the set $P_i$ is the set of vertices where every vertex in $P_i$ is colored with the same color and $P_i \subseteq X_i$. For each $j \in \{1,2,3,4,5\}$, the subgraph induced on each $X_j$ is a clique on $n+1$ vertices and the vertex $z$ is adjacent to every vertex in $X_j$. Finally $\displaystyle V(G_2) = \bigl ( \bigcup_{j \in [5]} X_j \bigr ) \cup \{z\}$.}
    \label{fig:construction}
\end{figure}

\begin{lemma}\label{lem:equivalence}
The graph $G_1$ has a vertex cover of size at most $k$ if and only if
\begin{equation*}
\vi(G_2)\leq n+k+2.
\end{equation*}
\end{lemma}

\begin{proof}
We may assume that $k\leq n$, since otherwise the instance is trivially a
YES-instance.

\textsf{(Only-if part)}
Suppose $G_1$ has a vertex cover $C\subseteq V(G_1)$ with $|C|\leq k$.
Let $S=C\cup\{z\}$. For any distinct $i,j\in\{1,2,3,4\}$, every edge
between $X_i$ and $X_j$ corresponds to an edge of $G_1$. Since $C$ is a
vertex cover of $G_1$, no such edge survives in $G_2-S$. Moreover, the
universal vertex $z$ is deleted.

Therefore, the components of $G_2-S$ are the nonempty sets
$X_i\setminus C$, for $i\in\{1,2,3,4\}$, together with the clique $X_5$.
Each $X_i\setminus C$ has at most $n+1$ vertices, while $X_5$ has exactly
$n+1$ vertices. Hence
\begin{equation*}
\mu_{G_2}(V(G_2)\setminus S)=n+1.
\end{equation*}
Since $|S|=|C|+1\leq k+1$, we obtain
\begin{equation*}
|S|+\mu_{G_2}(V(G_2)\setminus S)
\leq k+1+n+1
=n+k+2.
\end{equation*}
Thus $\vi(G_2)\leq n+k+2$.

\medskip
\textsf{(If part)}
Conversely, suppose $\vi(G_2)\leq n+k+2$. Then there exists a set
$S\subseteq V(G_2)$ such that
\begin{equation*}
|S|+\mu_{G_2}(V(G_2)\setminus S)\leq n+k+2.
\end{equation*}

\medskip
\noindent
\textit{Claim 1.} $z\in S$.

\begin{claimproof}
Suppose, for a contradiction, that $z\notin S$. Since $z$ is adjacent to
every other vertex of $G_2$, the graph $G_2-S$ is connected. Therefore,
\begin{equation*}
\mu_{G_2}(V(G_2)\setminus S)
=|V(G_2)|-|S|
=5(n+1)+1-|S|.
\end{equation*}
It follows that
\begin{equation*}
|S|+\mu_{G_2}(V(G_2)\setminus S)
=5(n+1)+1
=5n+6.
\end{equation*}
Since $k\leq n$, we have $n+k+2\leq 2n+2<5n+6$, a contradiction.
\end{claimproof}

We next show that we may assume
\begin{equation*}
S\setminus\{z\}\subseteq V(G_1).
\end{equation*}
Suppose that
$w\in S\setminus(V(G_1)\cup\{z\})$.
Then $w$ is one of the newly introduced vertices. If
$w\in X_i\setminus P_i$ for some $i\in\{1,2,3,4\}$, its only neighbour
outside $X_i$ is $z$. If $w\in X_5$, then all its neighbours belong to
$X_5\cup\{z\}$. Since $z\in S$, in either case all neighbours of $w$
outside $S$ lie in at most one component of $G_2-S$.

Hence, if $w$ is removed from $S$, restoring $w$ can increase the order of
the largest component by at most one. Thus
\begin{equation*}
\mu_{G_2}\bigl(V(G_2)\setminus(S\setminus\{w\})\bigr)
\leq
\mu_{G_2}(V(G_2)\setminus S)+1.
\end{equation*}
At the same time, the size of the deletion set decreases by one.
Therefore, replacing $S$ by $S\setminus\{w\}$ does not increase the value
$|S|+\mu_{G_2}(V(G_2)\setminus S)$. Repeating this operation for every
such vertex $w$, we may assume that
$S\setminus\{z\}\subseteq V(G_1)$.

Since no vertex of $X_5$ is deleted and $z\in S$, the clique $X_5$ is a
component of $G_2-S$ of order $n+1$. Therefore,
\begin{equation*}
\mu_{G_2}(V(G_2)\setminus S)\geq n+1.
\end{equation*}
Together with
\begin{equation*}
|S|+\mu_{G_2}(V(G_2)\setminus S)\leq n+k+2,
\end{equation*}
this gives $|S|\leq k+1$. Since $z\in S$, we obtain
\begin{equation*}
|S\setminus\{z\}|\leq k.
\end{equation*}

It remains to show that $S\setminus\{z\}$ is a vertex cover of $G_1$.
Suppose, for a contradiction, that some edge $uv\in E(G_1)$ has neither
endpoint in $S$. Since the colouring of $G_1$ is proper, there are distinct
indices $i,j\in\{1,2,3,4\}$ such that $u\in P_i$ and $v\in P_j$.
Let
\begin{equation*}
s_i=|S\cap P_i|
\qquad\text{and}\qquad
s_j=|S\cap P_j|.
\end{equation*}
Since no newly introduced vertex other than $z$ belongs to $S$, the sets
$X_i\setminus S$ and $X_j\setminus S$ have $n+1-s_i$ and $n+1-s_j$
vertices, respectively. Both sets induce cliques, and the edge $uv$
survives in $G_2-S$. Hence these two cliques lie in the same component of
$G_2-S$, and therefore
\begin{equation*}
\mu_{G_2}(V(G_2)\setminus S)
\geq (n+1-s_i)+(n+1-s_j).
\end{equation*}

Put $S'=S\setminus\{z\}$. Since $S'\subseteq V(G_1)$ and $z\notin
P_i\cup P_j$, we have
\begin{align*}
|S|+\mu_{G_2}(V(G_2)\setminus S)
&\geq |S|+2(n+1)-s_i-s_j\\
&=2n+3+|S'\setminus(P_i\cup P_j)|\\
&\geq 2n+3.
\end{align*}
However, $k\leq n$ implies $n+k+2\leq 2n+2$, contradicting
\begin{equation*}
|S|+\mu_{G_2}(V(G_2)\setminus S)\leq n+k+2.
\end{equation*}
Thus every edge of $G_1$ has at least one endpoint in
$S\setminus\{z\}$. Hence $S\setminus\{z\}$ is a vertex cover of $G_1$,
and its size is at most $k$. \hfill $\square$
\end{proof}

\begin{proof}[Proof of Theorem~\ref{thm:eth}]
Suppose, for a contradiction, that \textsc{Vertex Integrity} can be solved
in $2^{o(N)}N^{O(1)}$ time on graphs with $N$ vertices.

Let $(G_1,k)$ be an instance of \textsc{Vertex Cover}, where $G_1$ has
$n$ vertices and maximum degree at most three. Apply the construction
described above to obtain the graph $G_2$. By construction,
\begin{equation*}
N=|V(G_2)|=5(n+1)+1=O(n).
\end{equation*}
By Lemma~\ref{lem:equivalence}, $G_1$ has a vertex cover of size at most
$k$ if and only if
\begin{equation*}
\vi(G_2)\leq n+k+2.
\end{equation*}

Therefore, the assumed algorithm for \textsc{Vertex Integrity} would solve
\textsc{Vertex Cover} on $n$-vertex graphs of maximum degree at most three
in $2^{o(N)}N^{O(1)}=2^{o(n)}n^{O(1)}=2^{o(n)}$ 
time. This contradicts Lemma~\ref{lem:source}. Hence, unless ETH
fails, \textsc{Vertex Integrity} cannot be solved in
$2^{o(n)}n^{O(1)}$ time on $n$-vertex graphs.

\hfill $\square$
\end{proof}

\section{Exact Exponential Algorithm (Proof of Theorem~\ref{thm:algorithm})}\label{sec:algorithm}

In this section we present an exact exponential algorithm to compute the vertex integrity of a given graph.

We use the following reformulation of a characterization of Goddard and
Swart~\cite[Section~5]{goddard1990integrity}; a proof is included for
completeness.

\begin{lemma}\label{lem:deficit}
For a graph $G = (V,E)$ of order $n$,
  $$\vi(G)=n-\max\left\{|U|:U\subseteq V,
  \ \mu(V\setminus N[U])\ge\mu(U)\right\}.$$
\end{lemma}

\begin{proof}

Let
$M=\max\left\{|U|:U\subseteq V,\ 
\mu(V\setminus N[U])\geq \mu(U)\right\}$.
We show that $\vi(G)=n-M$.
The statement is immediate for $n=0$, so assume $n\geq1$.
Choose an integrity separator $S$ with $S\neq V$. Such a separator always
exists. Indeed, if $V$ is an integrity separator, then $\vi(G)=n$. Since
the empty set gives a value at most $n$, it also attains the optimum, and
hence $\varnothing$ is an integrity separator.

Let $C$ be a largest component of $G \setminus S$, and let
$U=V\setminus(S\cup V(C))$.
 Thus, $U$ is the union of all components of
$G \setminus S$ other than $C$. In particular,
$\mu(U)\leq |V(C)|$. Since different components of $G \setminus S$ have no edges
between them, no vertex of $C$ belongs to $N[U]$. Hence
$V(C)\subseteq V\setminus N[U]$, and therefore
\begin{equation*}
\mu(V\setminus N[U])\geq |V(C)|\geq \mu(U).
\end{equation*}
Thus $U$ is feasible in the definition of $M$. Moreover,
\begin{equation*}
n-|U|=|S|+|V(C)|=\vi(G).
\end{equation*}
It follows that $M\geq n-\vi(G)$, and hence
$\vi(G)\geq n-M$.

For the reverse inequality, let $U\subseteq V$ satisfy
$\mu(V\setminus N[U])\geq\mu(U)$, and set $q=\mu(U)$. If $q=0$, then
$U=\varnothing$, and $\vi(G)\leq n=n-|U|$.

Assume now that $q\geq1$. Since
$\mu(V\setminus N[U])\geq q$, some component of $G-N[U]$ contains at
least $q$ vertices. Choose a set $W$ of exactly $q$ vertices from this
component such that $G[W]$ is connected. Such a set exists by taking a
spanning tree of the component and repeatedly deleting leaves.

Since $W\subseteq V\setminus N[U]$, the sets $U$ and $W$ are disjoint and
there is no edge between them. Let
$S=V\setminus(U\cup W)$. Then $G \setminus S=G[U\cup W]$, and its components are
$G[W]$ together with the components of $G[U]$. The component $G[W]$ has
order $q$, while every component of $G[U]$ has order at most $q$.
Therefore $\mu(V\setminus S)=q$. Since $|W|=q$, we obtain
\begin{align*}
|S|+\mu(V\setminus S)
  &= n-|U|-|W|+q\\
  &= n-|U|.
\end{align*}
Hence $\vi(G)\leq n-|U|$ for every feasible set $U$. Taking $U$ of
maximum cardinality gives $\vi(G)\leq n-M$.

Combining the two inequalities gives $\vi(G)=n-M$.
\hfill $\square$
\end{proof}


A set $S\subseteq V(G)$ is irredundant if every vertex of $S$ has
neighbours in at least two distinct components of $G \setminus S$. We use the
following result of Gima et al.~\cite{gima2025structural}

\begin{lemma}[\cite{gima2025structural}]\label{lem:irredundant}
Every graph has an integrity separator that is irredundant.
\end{lemma}

Next, we prove some properties of such an irredundant integrity separator.

\begin{lemma}\label{lem:balanced}
Let $G$ be a non-empty graph and $S$ be an irredundant integrity separator, let $D_1,\ldots,D_m$ be the connected
components of $G \setminus S$, and let $h=\max_i|V(D_i)|$. The set of components can be
partitioned into two collections whose vertex unions $A$ and $B$ satisfy $\bigl||A|-|B|\bigr|\leq h$ and $|N(A)\cap N(B)|\geq |S|/2$.

\end{lemma}

\begin{proof}
Relabel the components so as to have $|D_1|\geq |D_2|\geq \cdots \geq |D_m|$. If $m$ is
odd, add one empty dummy component, and pair consecutive components. For
each pair, place one component in $A$ and the other in $B$. Put
$\delta_i=|D_{2i-1}|-|D_{2i}|\geq 0$. For every orientation of the pairs,
$\bigl||A|-|B|\bigr|\leq \sum_i\delta_i$. Moreover,
$\sum_i\delta_i=|D_1|-|D_2|+|D_3|-|D_4|+|D_5|-\cdots = |D_1|-(|D_2|-|D_3|+|D_4|-|D_5|+ \cdots)  \leq |D_1|\leq h$. Thus $\bigl||A|-|B|\bigr|\leq h$.

Next, we orient the pairs independently and uniformly at random. Consider a vertex $s\in S$.
Since $S$ is irredundant, $s$ has neighbours in at least two distinct
components of $G \setminus S$. If one pair contains two components adjacent to $s$,
then $s$ has a neighbour on each side for every orientation. Otherwise, if
$r\geq 2$ distinct pairs contain components adjacent to $s$, the probability
that all such components are put on the same side is
$2(1/2)^r=2^{1-r}\leq 1/2$. Hence,
$\Pr[s\in N(A)\cap N(B)]\geq 1/2$. By linearity of expectation, some
orientation satisfies $|N(A)\cap N(B)|\geq |S|/2$.  \hfill $\square$
\end{proof}

The random orientations are used only to prove the existence of such a partition. Algorithm \ref{alg:vi} below is deterministic: it enumerates all relevant choices of $A$, so it need not construct the partition from the unknown separator.

\begin{lemma}\label{lem:threecases}
Let $G$ be a non-empty graph and $S$ be an irredundant integrity separator. Put $k=|S|$, let $h$ be
the order of a largest component of $G \setminus S$, and let $d=n-k-h$.
For an integer $t$, if $k>t$ and $d>t$ then the set $V\setminus S$ can be partitioned into two sets $A$ and $B$ such that they are anticomplete and $\max\{|A|,|B|\}<n-3t/2$ and
$|V\setminus N[A]|<\frac{n}{2}-\frac{t}{4}$.
\end{lemma}

\begin{proof}
Since $k,d>t$ we have $h=n-k-d<n-2t$. By Lemma~\ref{lem:balanced}, the components of $G \setminus S$ can be partitioned into anticomplete sets $A,B$ with $|A|+|B|=n-k<n-t$ and
$\bigl||A|-|B|\bigr|\leq h<n-2t$, where $\max\{|A|,|B|\}<\bigl((n-t)+(n-2t)\bigr)/2=n-3t/2$.

Every vertex of the irredundant separator has a neighbour in at least two
components of $G \setminus S$, hence in $A\cup B$, so $N[A]\cup N[B]=V$. As $A$ and
$B$ are anticomplete, $N[A]\cap N[B]=N(A)\cap N(B)$ has order at least
$k/2>t/2$ by Lemma~\ref{lem:balanced}, so
$|N[A]|+|N[B]|=n+|N[A]\cap N[B]|>n+t/2$. One of the two closed
neighbourhoods therefore has order greater than $(n+t/2)/2$; considering that
side as $A$ gives $|V\setminus N[A]|<n-(n+t/2)/2=n/2-t/4$.
\qed
\end{proof}

The following is a direct corollary of the above-stated lemma where we ascertain the value of $t$ satisfying the constraints mentioned in Lemma \ref{lem:balanced}.

\begin{corollary}\label{cor:threshold}
Let $S$ be an irredundant integrity separator, and let $t$ be an integer
such that $t\geq\lceil2n/5\rceil$. Then at least one of the following holds:
\begin{enumerate}
    \item[(1)] $S$ has at most $t$ vertices;
    
    \item[(2)] if $C$ is a largest component of $G \setminus S$, then the union of
    all components of $G \setminus S$ other than $C$ has at most $t$ vertices;
    
    \item[(3)] $V\setminus S$ can be partitioned into two anticomplete sets
    $A$ and $B$ such that $|A|,|B|\leq t$, and for one of the two sets,
    say $A$, we have $|V\setminus N[A]|\leq t$.
\end{enumerate}
\end{corollary}

\begin{proof}
Put $k=|S|$, let $h$ be the order of a largest component of $G \setminus S$, where $G$ is a non-empty graph
and put $d=n-k-h$.
Suppose that (1) and (2) do not hold. Then $k>t$ and $d>t$, so
Lemma~\ref{lem:threecases} applies. Since $t\geq 2n/5$,
\[
n-\frac{3t}{2}\leq t
\qquad\text{and}\qquad
\frac{n}{2}-\frac{t}{4}\leq t.
\]
Hence the sets $A$ and $B$ given by Lemma~\ref{lem:threecases} satisfy
$|A|,|B|\leq t$ and, for one of them, say $A$,
$|V\setminus N[A]|\leq t$. Thus (3) holds.
\hfill $\square$
\end{proof}

Let $S$ be an irredundant integrity separator. If $|S|\leq t$, then
enumerating all sets $S'\subseteq V$ with $|S'|\leq t$ and evaluating
$|S'|+\mu(V\setminus S')$ finds $\vi(G)$, since $S$ itself is among the
enumerated sets.

Suppose instead that the union of all components of $G \setminus S$ other than one
largest component has at most $t$ vertices. Let $U$ denote this union.
Then $|U|\leq t$, and Lemma~\ref{lem:deficit} gives
$\mu(V\setminus N[U])\geq\mu(U)$ and $\vi(G)=n-|U|$. Hence, by enumerating
all sets $U'\subseteq V$ with $|U'|\leq t$ satisfying
$\mu(V\setminus N[U'])\geq\mu(U')$ and taking the minimum value of
$n-|U'|$, we again obtain $\vi(G)$.

If neither of these two cases occur, the third alternative of
Corollary~\ref{cor:threshold} applies. The next lemma handles this remaining
case.

To handle the third alternative of Corollary~\ref{cor:threshold}, we enumerate sets $A$ with $|A|\leq t$ and $|R_A|\leq t$, where $R_A=V\setminus N[A]$. For each such $A$, every set $B\subseteq R_A$ is anticomplete to $A$, so deleting $V\setminus(A\cup B)$ gives objective value $n-|A|-|B|+\max\{\mu(A),\mu(B)\}$. Writing $q=\mu(A)$, finding the best choice of $B$ therefore amounts to maximizing $|B|-\max\{q,\mu(B)\}$ over all subsets of $R_A$.

Let $\mathcal{F}_t={X\subseteq V:|X|\leq t}$. For each $q\in{0,\ldots,t}$ and $R\in\mathcal{F}t$, define
\[
F_q(R)=\max_{B\subseteq R}
\bigl(|B|-\max\{q,\mu(B)\}\bigr).
\]
For a set $A\in\mathcal{F}_t$ such that
$R_A=V\setminus N[A]\in\mathcal{F}_t$, let $q=\mu(A)$. The minimum value obtained by choosing a set $B\subseteq R_A$ is $n-|A|-F_q(R_A)$.

If $A$ is the set satisfying the third alternative of
Corollary~\ref{cor:threshold} for an irredundant integrity separator,
then this value equals $\vi(G)$. The following lemma gives a recurrence
for computing $F_q(R)$.

\begin{lemma}\label{lem:dp}
For every $q\in\{0,\ldots,t\}$, we have $F_q(\emptyset)=-q$, and every
nonempty $R\in\mathcal{F}_t$ satisfies
\begin{equation}
F_q(R)=
\max\left\{ \bigl(|R|-\max\{q,\mu(R)\} \bigr), \bigl (\max_{v\in R}F_q(R\setminus\{v\} \bigr ) \right\}.
\label{eq:subset-recurrence}
\end{equation}
Every predecessor $R\setminus\{v\}$ belongs to $\mathcal F_t$ and has smaller cardinality. Consequently, for each fixed $q$, all values $F_q(R)$ can be computed in increasing order of $|R|$, using polynomial work per state.

Moreover, let $A\in\mathcal{F}_t$, $q=\mu(A)$, and
$R_A=V\setminus N[A]\in\mathcal{F}_t$. Then
\begin{equation}
\min_{B\subseteq R_A}
\left(
|V\setminus(A\cup B)|+\mu(A\cup B)
\right)
=n-|A|-F_q(R_A).
\label{eq:separator-table}
\end{equation}
\end{lemma}

\begin{proof}
The equality $F_q(\emptyset)=-q$ follows directly from the definition. Now let
$R\neq\emptyset$, and let $B\subseteq R$ attain $F_q(R)$. If $B=R$, its value
is the first term of~\eqref{eq:subset-recurrence}. Otherwise, choose
$v\in R\setminus B$. Since $B\subseteq R\setminus\{v\}$, its value is at most
$F_q(R\setminus\{v\})$. Conversely, the first term is realized by $B=R$, and
every other term is realized by a subset of $R\setminus\{v\}$. This proves the
recurrence. Since every predecessor has one fewer vertex, the stated evaluation
order follows.

For the second assertion, every $B\subseteq R_A$ is disjoint from and
anticomplete to $A$. Hence $\mu(A\cup B)=\max\{q,\mu(B)\}$, and therefore $|V\setminus(A\cup B)|+\mu(A\cup B) =n-|A|-\bigl(|B|-\max\{q,\mu(B)\}\bigr).$ Minimizing over $B\subseteq R_A$ and applying the definition of $F_q(R_A)$
proves~\eqref{eq:separator-table}. \qed
\end{proof}

Algorithm~\ref{alg:vi} combines the three cases above. It computes
$\vi(G)$, and an integrity separator can be recovered by retaining the
corresponding choices. For every set $X\in\mathcal F_t$, we precompute
$\mu(X)$ and $N[X]$. Whenever the algorithm requires $\mu(Y)$ for a set
$Y\notin\mathcal F_t$, such as $V\setminus S$ or
$V\setminus N[U]$, we compute it in polynomial time by finding the
connected components of $G[Y]$.

\begin{algorithm}[H]
\caption{Exact Vertex Integrity}\label{alg:vi}
\begin{algorithmic}[1]
\Require An undirected graph $G=(V,E)$ with $n=|V|$
\Ensure $\vi(G)$
\State $t\gets\lceil2n/5\rceil$ and
       $\mathcal F_t\gets\{X\subseteq V:|X|\leq t\}$
\State Precompute $\mu(X)$ and $N[X]$ for every $X\in\mathcal F_t$
\State Whenever $Y\notin\mathcal F_t$, compute $\mu(Y)$ when needed by finding the connected components of $G[Y]$
\State $\mathit{ans}\gets n$
\ForAll{$S\in\mathcal F_t$}
  \State $\mathit{ans}\gets
  \min\{\mathit{ans},|S|+\mu(V\setminus S)\}$
\EndFor
\ForAll{$U\in\mathcal F_t$}
  \If{$\mu(V\setminus N[U])\ge\mu(U)$}
    \State $\mathit{ans}\gets\min\{\mathit{ans},n-|U|\}$
  \EndIf
\EndFor
\For{$q=0,1,\ldots,t$}
  \State Compute $F_q(R)$ for all $R\in\mathcal F_t$ using
  Eq.~\eqref{eq:subset-recurrence}
  \ForAll{$A\in\mathcal F_t$ with $\mu(A)=q$}
    \State $R\gets V\setminus N[A]$
    \If{$|R|\leq t$}
      \State $\mathit{ans}\gets
      \min\{\mathit{ans},n-|A|-F_q(R)\}$
    \EndIf
  \EndFor
\EndFor
\State \Return $\mathit{ans}$
\end{algorithmic}
\end{algorithm}

We can now complete the proof of Theorem~\ref{thm:algorithm}.

\exactalgorithm*

\begin{proof}
The assertion is immediate when $n=0$. Therefore, assume that $n\geq1$. Every value considered by the algorithm is attained by a vertex-deletion set. This is immediate in the first phase. In the second
phase, it follows from Lemma~\ref{lem:deficit}, and in the third phase,
from Lemma~\ref{lem:dp}. Hence, the algorithm never returns a value smaller
than $\vi(G)$.

For the reverse inequality, let $S$ be an irredundant integrity separator,
which exists by Lemma~\ref{lem:irredundant}. Put $k=|S|$, let $h$ be the
order of a largest component of $G \setminus S$, and let $d=n-k-h$. Thus
$\vi(G)=k+h=n-d$.

If $k\leq t$, the first phase enumerates $S$ and obtains $\vi(G)$.
Suppose $d\leq t$, and let $U$ be the union of all components of $G \setminus S$
except one largest component $C$. Then $|U|=d$, $\mu(U)\leq h$, and
$N[U]\subseteq U\cup S$. Hence $C$ survives in $G-N[U]$, so
$\mu(V\setminus N[U])\geq h\geq\mu(U)$. The second phase therefore
considers $U$ and obtains
$n-|U|=n-d=\vi(G)$.

It remains to consider $k>t$ and $d>t$. By
Corollary~\ref{cor:threshold}, there is an anticomplete partition $A,B$ of
$V\setminus S$ such that $|A|\leq t$ and
$R_A=V\setminus N[A]$ has size at most $t$. Since $B\subseteq R_A$, the
third phase considers $A$ with $q=\mu(A)$. By
Eq.~\eqref{eq:separator-table}, the value recorded for $A$ is at most
\begin{align*}
n-|A|-\bigl(|B|-\max\{\mu(A),\mu(B)\}\bigr)
&=n-|A|-|B|+h\\
&=k+h=\vi(G),
\end{align*}
where $\max\{\mu(A),\mu(B)\}=h$ because $A$ and $B$ form an
anticomplete partition of $V\setminus S$. Thus the algorithm returns
exactly $\vi(G)$.

Let
$M(n,t)=\sum_{i\leq t}\binom{n}{i}=|\mathcal F_t|$.
Each phase uses polynomial work per set in $\mathcal F_t$, and for each
$q$ the recurrence of Lemma~\ref{lem:dp} uses $O(n)$ transitions per
state. Thus the running time is $O^*(M(n,t))$. The tables for different
values of $q$ are computed separately, so the space usage is also
$O^*(M(n,t))$.

For $t=\lceil2n/5\rceil$, the entropy bound gives $M(n,t)=O^*\left(2^{H(2/5)n}\right)$, where $H(x)=-x\log_2x-(1-x)\log_2(1-x)$. Since
$2^{H(2/5)}=5/(2^{2/5}3^{3/5})=1.96013\ldots<1.9602$, the time and space complexity bounds are $O(1.9602^n)$.

An integrity separator can be recovered within the same bounds. Retain the
best set in the first phase. In the second, for the selected $U$, choose a component of $G-N[U]$ having at least $\mu(U)$ vertices and,
within that component, choose a connected set $W$ of order $\mu(U)$ and return
$V\setminus(U\cup W)$; if $\mu(U)=0$, return $V$. In the third, recover an
optimizing set $B$ from the table $F_{\mu(A)}$ and return
$V\setminus(A\cup B)$.

\hfill $\square$
\end{proof}

\section{Conclusion}\label{sec:conclusion}

We gave a deterministic $O(1.9602^n)$-time exact algorithm for computing
the vertex integrity of an $n$-vertex graph, thereby improving over the
straightforward $O^*(2^n)$ enumeration. We also proved that, unless ETH
fails, \textsc{Vertex Integrity} cannot be solved in $2^{o(n)}$ time, using
a reduction with linear vertex blow-up from \textsc{Vertex Cover} on graphs of maximum
degree at most three.

Several questions remain open. Can the base $1.9602$ be improved further?
In particular, stronger bounds in the third case of
Corollary~\ref{cor:threshold}, or a different way of handling the
corresponding subsets, may lead to a faster exact algorithm. It is also
open whether vertex integrity can be computed in $O^*(c^n)$ time using
polynomial space for some constant $c<2$. Finally, narrowing the gap
between the ETH lower bound and the present upper bound remains an
interesting direction for future work.

\bibliographystyle{splncs04}
\bibliography{ref}

@article{BarefootEntringerSwart1987,
  author  = {Barefoot, Curtis A. and Entringer, Roger and Swart, Henda},
  title   = {Vulnerability in Graphs-A Comparative Survey},
  journal = {Journal of Combinatorial Mathematics and Combinatorial Computing},
  volume  = {1},
  pages   = {13--22},
  year    = {1987}
}

@article{bagga1992survey,
  title={A survey of integrity},
  author={Bagga, Kunwarjit S and Beineke, Lowell W and Goddard, Wayne D and Lipman, Marc J and Pippert, Raymond E},
  journal={Discrete Applied Mathematics},
  volume={37},
  pages={13--28},
  year={1992},
  publisher={Elsevier}
}

@article{clark1987computational,
  title={Computational complexity of integrity},
  author={Clark, Lane H and Entringer, Roger C and Fellows, Michael R},
  journal={J. Combin. Math. Combin. Comput},
  volume={2},
  pages={179--191},
  year={1987}
}

@article{fellows1989immersion,
  title={The immersion order, forbidden subgraphs and the complexity of network integrity},
  author={Fellows, Michael R and Stueckle, Sam},
  journal={J. Combin. Math. Combin. Comput},
  volume={6},
  number={1},
  pages={23--32},
  year={1989}
}

@article{drange2016computational,
  title={On the computational complexity of vertex integrity and component order connectivity},
  author={Drange, P{\aa}l Gr{\o}n{\aa}s and Dregi, Markus and van’t Hof, Pim},
  journal={Algorithmica},
  volume={76},
  number={4},
  pages={1181--1202},
  year={2016},
  publisher={Springer}
}

@article{gima2025structural,
  title={Structural parameterizations of vertex integrity},
  author={Gima, Tatsuya and Hanaka, Tesshu and Kobayashi, Yasuaki and Murai, Ryota and Ono, Hirotaka and Otachi, Yota},
  journal={Theoretical Computer Science},
  volume={1024},
  pages={114954},
  year={2025},
  publisher={Elsevier}
}

@article{casel2026combining,
  title={Combining crown structures for vulnerability measures},
  author={Casel, Katrin and Friedrich, Tobias and Niklanovits, Aikaterini and Simonov, Kirill and Zeif, Ziena},
  journal={Algorithmica},
  volume={88},
  number={1},
  pages={9},
  year={2026},
  publisher={Springer}
}

@inproceedings{lee2017partitioning,
  title={Partitioning a graph into small pieces with applications to path transversal},
  author={Lee, Euiwoong},
  booktitle={Proceedings of the Twenty-Eighth Annual ACM-SIAM Symposium on Discrete Algorithms},
  pages={1546--1558},
  year={2017},
  organization={SIAM}
}

@article{impagliazzo2001complexity,
  title={On the complexity of k-SAT},
  author={Impagliazzo, Russell and Paturi, Ramamohan},
  journal={Journal of Computer and System Sciences},
  volume={62},
  number={2},
  pages={367--375},
  year={2001},
  publisher={Elsevier}
}

@article{bacso2019subexponential,
  title={Subexponential-Time Algorithms for Maximum Independent Set in $P_t$-Free and Broom-Free Graphs: G. Bacs{\'o} et al.},
  author={Bacs{\'o}, G{\'a}bor and Lokshtanov, Daniel and Marx, D{\'a}niel and Pilipczuk, Marcin and Tuza, Zsolt and Van Leeuwen, Erik Jan},
  journal={Algorithmica},
  volume={81},
  number={2},
  pages={421--438},
  year={2019},
  publisher={Springer}
}

@article{goddard1990integrity,
  title={Integrity in graphs: bounds and basics},
  author={Goddard, Wayne and Swart, Henda C},
  journal={J. Combin. Math. Combin. Comput},
  volume={7},
  pages={139--151},
  year={1990}
}

@article{gima2022exploring,
  title={Exploring the gap between treedepth and vertex cover through vertex integrity},
  author={Gima, Tatsuya and Hanaka, Tesshu and Kiyomi, Masashi and Kobayashi, Yasuaki and Otachi, Yota},
  journal={Theoretical Computer Science},
  volume={918},
  pages={60--76},
  year={2022},
  publisher={Elsevier}
}

@InProceedings{hanaka2024parameterized,
  author    = {Tesshu Hanaka and Michael Lampis and Manolis Vasilakis and Kanae Yoshiwatari},
  title     = {Parameterized Vertex Integrity Revisited},
  booktitle = {49th International Symposium on Mathematical Foundations of Computer Science (MFCS 2024)},
  volume    = {306},
  pages     = {58:1--58:14},
  year      = {2024}
}

\end{document}